%% file: sample-acmcp.tex
\documentclass[sigconf]{acmart}
\usepackage{algorithm}
\usepackage{algorithmic}
\usepackage{xcolor}

\usepackage{tabularx} 
\usepackage{listings}
\usepackage{subcaption}

\setcopyright{none}

\AtBeginDocument{%
  }

\begin{document}

\title{Enhancing LLM-Based Test Generation by Eliminating Covered Code}



\author{Weizhe Xu}
\affiliation{%
  \institution{University of Notre Dame}
  \city{Notre Dame, IN}
  \country{USA}
}
\email{wxu3@nd.edu}

\author{Mengyu Liu}
\affiliation{%
  \institution{Washington State University}
  \city{Richland, WA}
  \country{USA}
}
\email{mengyu.liu@wsu.edu}

\author{Fanxin Kong}
\affiliation{%
  \institution{University of Notre Dame}
  \city{Notre Dame, IN}
  \country{USA}
}
\email{fkong@nd.edu}

\input{sections/0_abstract}

\maketitle

\input{sections/1_introduction}
\input{sections/2_relatedwork}
\input{sections/4_approach}
\input{sections/5_experiments}
\input{sections/7_conclusion}


\bibliographystyle{ACM-Reference-Format}
\bibliography{sample-base}

\newpage

\input{sections/8_appendix}

\end{document}

%% file: sections/0_abstract.tex
\begin{abstract}
Automated test generation is essential for software quality assurance, with coverage rate serving as a key metric to ensure thorough testing.
Recent advancements in Large Language Models (LLMs) have shown promise in improving test generation, particularly in achieving higher coverage.
However, while existing LLM-based test generation solutions perform well on small, isolated code snippets, they struggle when applied to complex methods under test.
To address these issues, we propose a scalable LLM-based unit test generation method. 
Our approach consists of two key steps.
The first step is context information retrieval, which uses both LLMs and static analysis to gather relevant contextual information associated with the complex methods under test.
The second step, iterative test generation with code elimination, repeatedly generates unit tests for the code slice, tracks the achieved coverage, and selectively removes code segments that have already been covered.
This process simplifies the testing task and mitigates issues arising from token limits or reduced reasoning effectiveness associated with excessively long contexts.
Through comprehensive evaluations on open-source projects, our approach outperforms state-of-the-art LLM-based and search-based methods, demonstrating its effectiveness in achieving high coverage on complex methods.
\end{abstract}

\keywords{Unit Test Generation, Large Language Models, Coverage}

%% file: sections/1_introduction.tex
\section{Introduction}

Achieving high code coverage in automated test generation is a critical yet challenging problem in software quality assurance, as insufficient coverage often leaves bugs undetected.
Traditional approaches—such as search-based~\cite{fraser2011evosuite, lukasczyk2022pynguin}, constraint-based~\cite{cadar2008klee, godefroid2005dart}, and random-based~\cite{pacheco2007randoop, 
maciver2019hypothesis} methods—frequently fall short when facing complex methods under test, as their heuristic explorations cannot efficiently handle the vast and intricate execution spaces.

Recently, large language models (LLMs) have presented new opportunities for improving the coverage rate of unit test generation.
Studies such as ChatTester~\cite{yuan2023no} and ChatUniTest~\cite{chen2024chatunitest} have evaluated LLMs, such as GPT-3.5, demonstrating their ability to achieve a higher coverage rate compared to conventional techniques like Evosuite~\cite{fraser2011evosuite}. 
However, existing LLM-based test generation methods work well on simple code snippets but struggle with complex functions (cyclomatic complexity~\cite{mccabe1976complexity} $>$ 10), which pose moderate or higher risk~\cite{sanusi2020development, yang2024enhancing}. 
Leveraging LLMs to generate high-coverage test code for complex methods is crucial for moving LLM-driven testing from toy examples to practical, production-level applications.


This is a non-trivial task due to two challenging issues.
First, current LLMs have limited reasoning capabilities.
Their performance declines on complex methods involving deeply nested conditions and multiple execution paths, often resulting in low coverage.
Second, prompt efficiency is a major concern.
LLMs face strict token limits (e.g., 16K tokens in GPT-4o), making it difficult to include all relevant context for complex methods.
Consequently, enhancing prompt efficiency is essential to avoid failures due to token overflow.
It is also important for improving the effectiveness of the LLM.
While some approaches~\cite{yuan2023no, chen2024chatunitest} try to pack more context into prompts, excessive or irrelevant information can introduce noise and hinder reasoning~\cite{shi2023large, zhou2024can, jin2025end}.
Multi-turn prompting may help to increase the reasoning capabilities, but it increases redundancy and token usage. 
Efficient prompting is thus critical to reduce noise, stay within limits, and improve output quality.

To address these challenges, we propose an LLM-based unit test generation framework for Python that incrementally eliminates already-covered code using a divide-and-conquer strategy.
Our method consists of two key steps.
(1) \textbf{Context information retrieval}: we collect and summarize all relevant code dependencies for the target unit, including external packages, helper functions, and method definitions.
(2) \textbf{Iterative test generation with code elimination}: the LLM repeatedly generates and refines test cases while the framework systematically removes code segments that have already been covered, ensuring that such removals do not alter the execution paths of any uncovered lines. After each iteration, this pruning produces a simplified program slice that serves for the next round of test generation.
By progressively reducing redundant code, our method keeps the LLM's attention focused on the remaining uncovered parts, narrows the reasoning scope, simplifies the test generation task, and mitigates performance degradation caused by irrelevant or redundant context.

The main contributions of our study are as follows.
\begin{itemize}
    \item To improve coverage performance on complex methods, we propose a fully automated LLM-based unit test generation approach tailored for Python.
    \item The approach employs an iterative unit test generation approach that enhances both efficiency and effectiveness by systematically eliminating already covered code.
    \item We conduct a comprehensive evaluation across multiple open-source projects using three LLM-based unit test generation method and one classical search-based method, demonstrating the effectiveness of our method in comparison to three state-of-the-art LLM-based approaches.
\end{itemize}


%% file: sections/2_relatedwork.tex
\section{Related Works}~\label{related}
In this section, we provide a brief overview of unit test generation methods.
A ``unit" typically refers to a function, method, or procedure within the code.
Unit test generation aims to automatically create test cases that validate the correctness of software by executing different code paths. 
Developers hope the generated unit tests can cover as many branches and lines as possible.

\subsection{Conventional Unit Test Generation Methods}
Researchers have explored various approaches to achieve high coverage, including search-based software testing (SBST)~\cite{lukasczyk2022pynguin, fraser2011evosuite}, symbolic execution~\cite{chen2013state, 10.1145/1064978.1065036}, and deep-learning-based test generation~\cite{tufano2020unit}.
SBST methods employ strategies like evolutionary algorithms in Evosuite~\cite{fraser2011evosuite} and coverage-directed random testing in Randoop~\cite{pacheco2007randoop}. 
However, their effectiveness is limited in complex software due to the vast search space~\cite{mcminn2011search}.
Symbolic execution tools, such as Dart~\cite{10.1145/1064978.1065036}, systematically explore feasible execution paths but suffer from scalability issues due to path explosion~\cite{xie2009fitness}.
Deep learning models generate test cases directly from code and offer scalability by leveraging human-written tests across diverse datasets~\cite{tufano2020unit}. 
While they produce varied inputs beyond SBST’s capabilities, they often generate code that is non-compliant or fails to execute~\cite{yang2022survey}.

\subsection{LLM-based Unit Test Generation Methods}
In recent years, LLMs have demonstrated their potential to effectively handle code-related tasks, such as code generation, code repair, etc.
Some works have already tried to leverage LLMs for unit test generation.
ChatTester~\cite{yuan2023no} and ChatUniTest~\cite{chen2024chatunitest} show that the LLM can outperform the SBST methods.
CODAMOSA~\cite{lemieux2023codamosa} integrates LLMs with SBST methods, leveraging LLMs to assist when SBST methods can no longer effectively explore further.
TELPA~\cite{yang2024enhancing} utilizes carefully designed prompts to assist LLMs in generating unit test cases that cover hard-to-cover branches.
SymPrompt~\cite{ryan2024code}, drawing inspiration from symbolic execution, leverages execution paths as structural guides for the LLM.

Although recent methods~\cite{yuan2023no, chen2024chatunitest, lemieux2023codamosa, pizzorno2024coverup} improve LLM-based test generation by providing more context, they overlook the scalability challenge for complex methods.
Scalability enlarges the analysis scope, introduces noise, and strains LLMs’ limited generation capacity, while also risking token overflow in multi-turn interactions.
Consequently, existing approaches still struggle to achieve satisfactory coverage for complex real-world methods.
HITS~\cite{wang2024hits} attempts to address this issue by splitting the method under test.
However, this approach relies on the LLM to perform a slicing.
Due to hallucination, it may miss important cross-slice information or introduce irrelevant or even new code generated by the LLM itself.
In addition, when dealing with larger code under test, a single round of splitting may not be sufficient.

In contrast, our approach fundamentally addresses this issue by eliminating already covered non-essential code.
Our method is based on static analysis and removes only the lines that are irrelevant to the execution of the uncovered lines, thereby avoiding information loss during cutting. 
Moreover, our approach performs multiple rounds of elimination for each code slice. 
After each elimination process, the LLM’s conversation history is cleared to minimize the impact of multi-turn interactions.

%% file: sections/4_approach.tex
\section{Methodology}~\label{method}
In this section, we present an overview of our proposed framework, as illustrated in Fig.~\ref{fig:overview}.
\begin{figure*}[htbp]
    \centering
    \includegraphics[width=\textwidth]{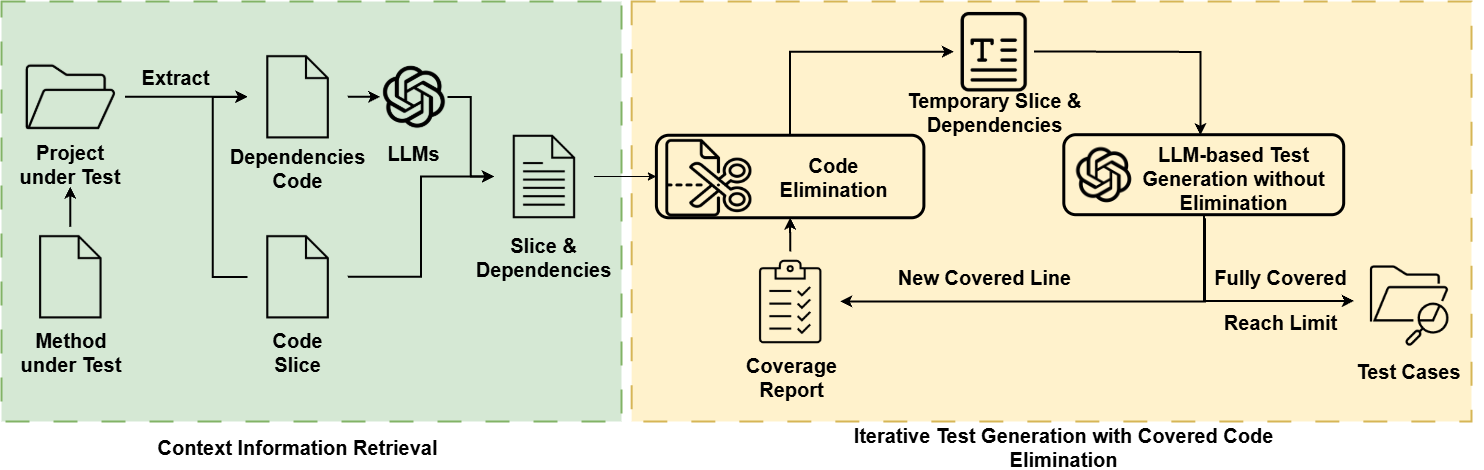}
    \caption{The Overview of Our Framework. Our approach consists of two steps: context information retrieval and iterative test generation with covered code elimination, which are illustrated by the green and orange dashed boxes in the figure, respectively.}
    \label{fig:overview}
\end{figure*}
The framework aims to generate high-coverage unit test cases for Python projects through two main steps: \textbf{context information retrieval} and \textbf{iterative test generation with covered code elimination}.



\subsection{Step I: Context Information Retrieval}
The context information retrieval step extracts relevant dependencies and summarizes them to support unit test generation for a complex target method under test, as shown in the green dashed box in Fig.~\ref{fig:overview}.
Given a target method, we first identify its internal dependencies, whose definitions are located in the same source file. 
This is done by analyzing intra-file function calls and internal variable references using static analysis tools such as the Abstract Syntax Tree (AST)~\cite{sun2023abstract}.
The definitions of these internal dependencies, together with the code of the target method itself, collectively form a basic code slice.
Next, we identify information about external dependencies by analyzing cross-module function calls and external variable references. 
For dependencies defined within the project, we collect their corresponding definition code to form a dependencies code file, similar to~\cite{wang2024hits, chen2024chatunitest, yang2024enhancing}

However, this dependency code file is not suitable for direct input to the LLM, as it often contains excessive irrelevant information and is typically much larger than the code slice itself. 
Including such content in the prompt as auxiliary information for unit test generation may consume a significant number of tokens and introduce substantial noise.
To ensure prompt efficiency and reduce noise, our method utilizes LLMs to summarize the dependency code file. 
Our method applies the one-shot prompting strategy to summarize the definition of each function collected in the dependencies code file.
Using a manually constructed example, we guide the LLM to summarize each function into its signature and a high-level description of its underlying logic.

The summarized dependencies, combined with the code slice, form the prefix for subsequent unit test generation.
This step is essential, as the content of local modules is typically project-specific and not encountered during LLM training, making it difficult for the model to infer behavior based solely on module names. 
Third-party libraries such as \textit{numpy} or \textit{pandas}, which LLMs are generally familiar with during the training process, are excluded from summarization. 
By focusing only on strongly relevant and compressed contextual information, this step helps reduce token overhead, minimize noise, and improve the accuracy and quality of the generated unit tests.

\subsection{Step II: Iterative Test Generation with Covered Code Elimination}
This is the core step, which takes the slice \& dependencies file obtained from the previous context information retrieval step as input and leverages the LLM to generate unit tests with high coverage rate.
This step can be divided into two components.
One is \textbf{LLM-based test generation for a code slice without elimination}, and the other is \textbf{code elimination}.

As shown in the yellow dashed box in Fig.~\ref{fig:overview}, the slice \& dependencies file is first sent to the code elimination component.
At this time, the coverage report contains no covered lines within the target unit. 
Therefore, the code elimination component does not perform any elimination, and the temporary slice \& dependencies file remains identical to the original slice \& dependencies file. 
This file is then passed to the LLM-based test generation without the elimination component to generate unit tests. 
If all lines in the target unit are covered, the resulting tests are added to the final test suite. 
If there are still uncovered lines, a coverage report is produced and forwarded to the code elimination component. 
This component eliminates portions of the code slice that are irrelevant to the uncovered lines, while aiming to minimize the impact on the execution of the uncovered lines.
Then a new temporary slice \& dependencies file is created. 
This process repeats until either all lines in the target method are covered or the LLM-based test generation without the elimination component can no longer generate tests that cover any of the remaining uncovered lines.

\begin{figure}[H]  
    \centering
    \includegraphics[width=\columnwidth]{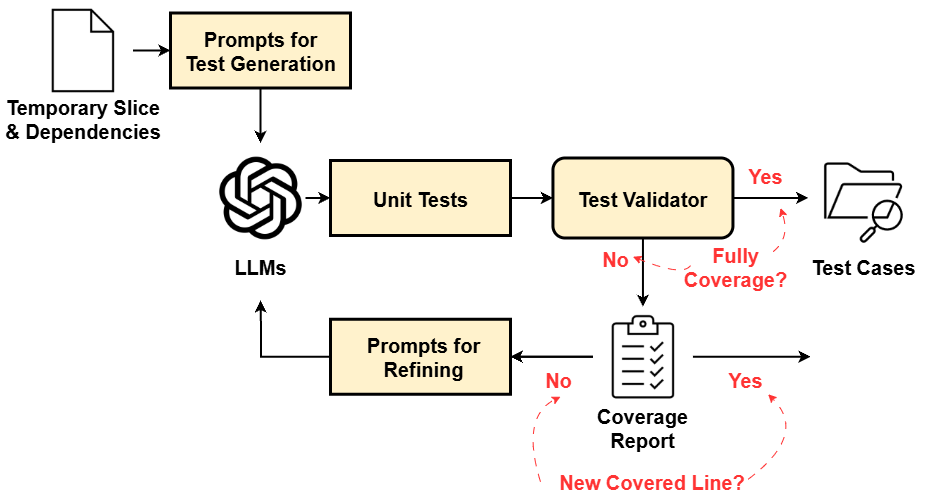} 
    \caption{Internal Structure of LLM-based Test Generation without Elimination Component.}
    \vspace{-2mm}
    \label{fig:iteration}
\end{figure}

\subsubsection{LLM-based Test generation without Elimination.}
This component receives the temporary slice \& dependencies file from the code elimination component and outputs the generated test cases along with the corresponding coverage report.
The detailed internal structure is illustrated in Fig.~\ref{fig:iteration}.

In the beginning, the temporary slice \& dependencies files are used to construct a prompt.
The format of the prompt is shown as Fig.~\ref{fig:format}.
In this format, the section enclosed in \{\{$\cdot$\}\} represents placeholders that are dynamically filled based on the code under test when generating each prompt.
For example, \{\{$code\_under\_test$\}\} will be replaced with the specific code slice.
\begin{figure}[h]  
    \centering
    \includegraphics[width=\columnwidth]{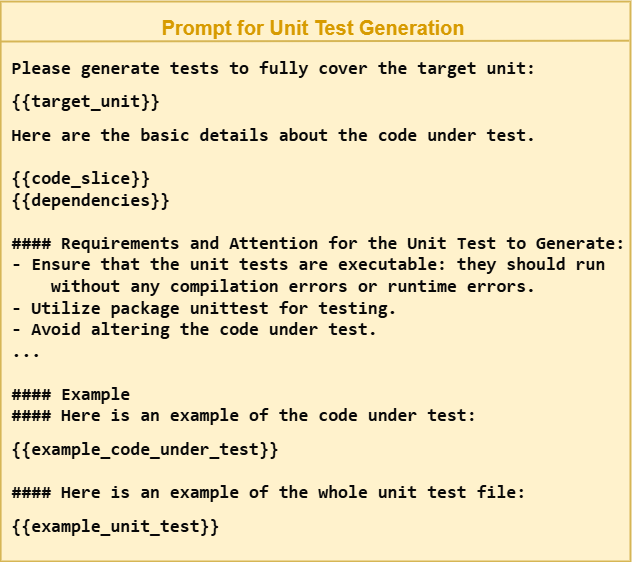} 
    \caption{The Format of the Prompt. \{\{$\cdot$\}\} serves as a placeholder, which will be replaced by variables during actual use.}
    \label{fig:format}
\end{figure}

We use this prompt to guide the LLM in generating unit tests, which are then evaluated by a test validator for coverage.
All previously generated tests are executed, and coverage is measured against the original Python file.
If full coverage is achieved, the validator outputs the complete test set; otherwise, it checks for newly covered lines.
If coverage improves, the updated report is forwarded to the code elimination component; if not, refinement prompts are generated to guide further test improvement.
Runtime errors from failed tests are also collected and incorporated into refinement prompts.
This process repeats until reaching a predefined iteration limit.

The overall process of LLM-based test generation without elimination can be represented as Algorithm.~\ref{alg:alg2}.
The return values 1, 0, and -1, respectively, indicate the following conditions: the target slice has been fully covered and the process terminates; new lines have been covered, processing another round of code elimination; and the iteration limit has been exceeded, leading to termination.
In Lines 1 and 2, the algorithm constructs a new prompt and starts a new dialogue with LLM.
Starting from Line 4, the algorithm periodically prompts the LLM to generate tests and validate them, and takes different actions based on the validation results.

\begin{algorithm}
\caption{LLM-based Tests Generation without Elimination}
\label{alg:alg2}
\begin{algorithmic}[1]
\REQUIRE $slice$, $limit$, $tests$, $uncov$
\ENSURE $tests$, $uncov$, $flag$
\STATE $prompt \leftarrow prompt\_constructor(slice, uncov)$
\STATE $llm \leftarrow$ new $LLM\_dialogue()$
\STATE $iteration \leftarrow 0$
\WHILE{$iteration \leq limit$}
    \STATE $test \leftarrow llm(prompt)$
    \STATE $tests.add(test)$
    \STATE $new\_uncov \leftarrow validator(tests)$
    \IF{$new\_uncov = [~]$}
        \RETURN $tests$, $new\_uncov$, $1$
    \ELSIF{$len(new\_uncov) < len(uncov)$}
        \RETURN $tests$, $new\_uncov$, $0$
    \ELSE
        \STATE $prompt \leftarrow refining(new\_uncov)$
    \ENDIF
    \STATE $uncov \leftarrow new\_uncov$
    \STATE $iteration \leftarrow iteration + 1$
\ENDWHILE
\RETURN $tests$, $uncov$, $-1$
\end{algorithmic}
\end{algorithm}

Through this iterative LLM invocation process, we maximize the utilization of the LLM's limited code-related capabilities to generate test cases for a slice \& dependencies file.
There are various ways to construct such prompts~\cite{yang2024enhancing, ryan2024code}, but prompt engineering is not our focus, and existing methods can be readily integrated into our framework.

\subsubsection{Code Elimination.}
This component takes the coverage analysis report, along with the original code slice \& dependencies file, as input and generates a new temporary slice \& dependencies file.
It systematically removes as many already covered lines as possible while preserving the original execution behavior, essentially decomposing a complex problem into simpler subproblems.

Specifically, all the uncovered liens need to be preserved.
Then, we construct a fine-grained control flow graph (CFG) of the original target method based on the AST.
It is a is a directed graph that represents all possible statement execution paths through a program, where nodes correspond to statements and edges indicate control flow between them.
From the perspective of the set of execution paths, our goal is to eliminate all execution paths that do not contain any uncovered line.
For each uncovered line, we perform a bidirectional breath-first search (BFS) on the fine-grained CFG.
All nodes reached during these traversals correspond to statements that must be preserved, as they represent the statements required by all execution paths that include the uncovered line.
Finally, it is also necessary to refine some retained statements. 
For example, in an $if...else...$ structure, if the $if$ branch becomes empty after pruning, the statements in the $else$ branch should be transformed into an $if~not$ condition to ensure syntactic correctness and preserve program logic.


\begin{algorithm}
\caption{Code Elimination}
\label{alg:alg1}
\begin{algorithmic}[1]
\REQUIRE $uncov$, $target$
\ENSURE $new\_target$
\STATE $preserve \leftarrow uncov$
\STATE $cfg \leftarrow construct\_CFG(target)$
\FOR{each $line$ in $uncov$}
    \STATE $necessities \leftarrow bidirect\_BFS(line, cfg)$
    \STATE $preserve.add(necessities)$
\ENDFOR
\STATE $new\_target \leftarrow fix(preserve, target)$
\RETURN $new\_target$
\end{algorithmic}
\end{algorithm}

The overall process of code elimination can be represented as Algorithm.~\ref{alg:alg1}.
The algorithm takes as input the set of uncovered lines ($uncov$) and the original target unit ($target$), and returns a refined target unit $new\_target$ with irrelevant code eliminated.
In Line 1, all uncovered lines are initially added to the set $preserve$, which indicates the lines that must be preserved.
In line 2, the algorithm uses the function $construct\_CFG(\cdot)$ to construct a fine-grained CFG of the original target method.
Lines 3-6 iterate over each uncovered line.
For each such line, bidirectional BFS is performed on the CFG to identify all its necessary lines that are required for execution.
These necessary lines are then added to the $preserve$ set.
In Line 7, the function $fix(\cdot)$ is invoked to reconstruct the new target unit by retaining only the lines in $preserve$ set.



\begin{proposition}[Behavior Preservation]
Let $P$ be the original target unit, $U$ the set of uncovered lines, and $P'$ the 
program produced by Algorithm~1. Assume that 
(1) the fine-grained CFG soundly captures all execution paths of $P$; 
(2) all data dependencies of lines in $U$ lie on execution paths leading to them; and 
(3) the reconstruction step performs only semantics-preserving local rewrites. 
Then, for any input on which $P$ reaches a line $l \in U$, the transformed program $P'$ 
reaches the corresponding line and observes the \emph{same values for all variables used at $l$} as in $P$.
\end{proposition}

\begin{proof}
In algorithm~\ref{alg:alg1}, bidirectional BFS on the fine-grained CFG collects all necessary statements that may occur on any execution path including each $l \in U$. 
By Assumption~(2), all statements that define or influence variables used at $l$ appear on such paths and are therefore preserved. 
Statements not in the preserved set never occur on paths including $l$.
So removing them cannot affect variable values at $l$. 
The reconstruction step applies only semantics-preserving structural rewrites, ensuring that preserved statements execute in the same order and with the same effects. 
Thus, for any execution of $P$ that includes an uncovered line $l$, $P'$ executes the same sequence of preserved statements in the same order, resulting in identical valuations of all variables used in $l$.
\end{proof}

In practice, code elimination may not perfectly maintain syntactic correctness or program logic.
Nevertheless, the essential information for test generation remains preserved in the new code \& dependencies slice.
Furthermore, all LLM-generated tests are evaluated on the original source file, and each elimination round is applied to the original target unit, preventing error accumulation.

After completing the process of code elimination, we obtain a smaller temporary slice \& dependencies file and apply it to the LLM-based tests generation without the elimination component.
This cycle repeats until the iteration limit in the LLM-based tests generation without the elimination component is reached, and the coverage rate remains unchanged.
During this iterative process, each execution of the code elimination not only reduces the token count of the code under test, but also enhances the prompt efficiency by narrowing the analysis scope for the LLM.
This focused context allows the LLM to focus on the relevant logic of the uncovered lines, thereby improving its overall performance.

%% file: sections/5_experiments.tex
\section{Experiments}~\label{exp}
In this section, we first present the setup of the experiment and then detail the results and analysis of the experiment.
\subsection{Experiment Setup}

\subsubsection{Benchmarks.}
We construct our dataset following established practices in prior influential work, including both search-based methods~\cite{lukasczyk2022pynguin} and LLM-based approaches~\cite{lemieux2023codamosa, yang2024enhancing}.
Our dataset is built under three criteria.
First, we include projects used in Pynguin~\cite{lukasczyk2022pynguin} and TELPA~\cite{yang2024enhancing}, as they serve as part of our baselines.
Second, we select projects from diverse domains with high GitHub star counts to ensure practical relevance.
Third, each project must contain 1–40 complex methods, defined as methods with cyclomatic complexity above 10 and more than 50 lines of code; these methods constitute our target units.
These criteria ensure that our dataset is both high-quality and representative, while maintaining a balance between diversity and computational budget.
The corresponding statistics are presented in Table.~\ref{tab:benchmarks}.
Each row is an open-source project.
The "Revision" column indicates the specific version of the project used in our evaluation. 
The "Domain" column represents the application area of the open-source project.
"\#MUTs" denotes the number of methods under test in each project.

\renewcommand{\arraystretch}{1.2}
\begin{table*}[t]
\small
\centering
\begin{tabular}{lllrrrrr}
\hline
\textbf{Project} & \textbf{Revision} & \textbf{Domain} & \textbf{\#MUTs} & \multicolumn{3}{c}{\textbf{\# Lines}}        \\ \cline{5-7}
                 &                   &                 &                 & \textbf{min} & \textbf{mean} & \textbf{max} &                   \\ \hline
dataclasses-json & dc63902  & Serialization     &   6    & 53      &    60    &   85    &          \\ 
docstring parser & 4951137  & Text Processing   &   36   & 51      &    93    &   260   &          \\ 
flutes           & 3b7c518  & Utility           &   4    & 50      &    70    &   104   &          \\ 
flutils          & df0f84e  & Utility           &  18    & 50      &    87    &   180   &          \\ 
mimesis          & b981966  & Data Processing   &   6    &    53   &     64   &      60 &          \\ 
sanic            & a64d7fe  & Microservices     &   20    &     52  &      189  &      89 &          \\ 
pytutils         & d7a37c0  & Utility           &   3    &     51  &      67  &    80   &          \\ 
thonny           & e019c1e  & IDE               &   20   &     50  &      144 &    78   &          \\ 
cookiecutter     & b445123  & Utility           &   10   &     52  &      172 &    101  &          \\ \hline
\end{tabular}
\caption{Characteristics of the benchmarks used in the evaluation.}
\label{tab:benchmarks}
\end{table*}

\subsubsection{Baselines.}
We select three state-of-the-art LLM-based unit test generation methods and one SBST approach as our baselines. 
The details are as follows:
\begin{itemize}
    \item \textbf{ChatUniTest (CUT)~\cite{chen2024chatunitest}:}
    This is a widely cited LLM-based automated unit test generation framework that employs a generation-validation-repair mechanism to mitigate errors produced by the LLM.
    \item \textbf{TELPA~\cite{yang2024enhancing}:}
    It is a novel LLM-based test generation method that tackles hard-to-cover branches.
    It refines ineffective tests as counter-examples and integrates program analysis results into the prompt to guide LLMs in generating more effective tests.
    \item \textbf{HITS~\cite{wang2024hits}:} 
        It is a method that enhances LLM-based unit test generation by slicing complex focal methods into smaller parts, allowing the LLM to generate tests incrementally. 
    \item \textbf{Pynguin~\cite{lukasczyk2022pynguin}:}
        It is a widely used search-based software testing tool for Python that leverages evolutionary algorithms to automatically generate unit tests aiming to maximize code coverage.
\end{itemize}

For TELPA, as the original paper does not release the source code, we reimplement the approach based on the descriptions provided.
For ChatUniTest and HITS, which are originally designed for generating unit tests in Java, we develop an adapted version tailored for Python.

\subsubsection{Experimental Parameters.}
We use OpenAI's GPT-4o~\cite{openai2024gpt4o} as the LLM in this work since this model strikes a good balance between the price and the performance in code-related works.
For the GPT-4o API, we use the default values for the temperature and token length limit parameters. 
Specifically, the temperature is set to 1.0, and the token length limit is 8096.
For each LLM-based method, the number of interaction rounds with the LLM in a dialogue is limited to 5.
For Pynguin, we set the max iteration number to 100.

\subsection{Coverage Rate Comparison}
We present and analyze the line coverage achieved when testing complex Python methods, which serves as the core evaluation of our approach. 
Due to space limitations, we report only line coverage, since prior studies~\cite{lemieux2023codamosa, yang2024enhancing, wang2024hits} have shown that line and branch coverage are strongly correlated across unit test generation methods.
Coverage is computed using the widely adopted tool \textit{Coverage}~\cite{coveragepy}.

Table~\ref{tab:line} summarizes the results, where each row corresponds to a project and each column to a test generation method. 
Bold values denote the best performance for each project. 
As shown in the table, \textbf{our method achieves the highest average line coverage among all baselines}, and attains the best result in 7 out of the 9 projects, demonstrating its effectiveness on complex target methods.

Overall, existing LLM-based methods struggle with the complex functions. 
For instance, TELPA achieves only 20.51\% coverage on the complex target function in \textit{docstring parser}, despite reporting 92.20\% coverage for the full project in its original paper. 
This significant drop highlights the scalability limitations of current LLM-based approaches when dealing with highly complex functions.

Moreover, TELPA and HITS do not consistently outperform each other, but both perform better than the basic ChatUniTest.
This suggests that the context enrichment used in TELPA and the lightweight LLM-guided decomposition adopted by HITS can be effective to some extent, but their overall impact remains limited compared to our method.

The SBST method, Pynguin, also performs poorly on complex functions, yielding 0.00\% coverage on six projects. 
This occurs because its evolutionary search often cannot produce valid inputs that satisfy the strict preconditions, deep control flows, or large input domains of these complex methods. 
As a result, the generated tests fail to reach the entry paths of the target functions, leading to no statements being covered.
This reflects the superiority of LLM-based test generation methods.

\renewcommand{\arraystretch}{1.2}
\begin{table}[t]
\centering
\small
\begin{tabular}{@{}lrrrrr@{}}
\hline
\textbf{Project} & \textbf{Our} & \textbf{CUT} & {\textbf{TELPA}} & \textbf{HITS} & \textbf{Pynguin }\\ \hline
dataclasses-json &   24.27\%    &   \textbf{25.12\%}     &   16.50\%        &    15.53\%  & 16.57\% \\
docstring parser &   \textbf{40.71\%}    &   22.64\%   &   20.51\%        &    33.80\%   & 16.49\%\\
flutes           &   \textbf{49.34\%}    &   45.80\%   &   49.34\%        &    48.68\%   & 0.00\% \\
flutils          &   \textbf{41.04\%}    &   16.22\%   &   35.07\%        &    20.90\%   & 15.15\%\\
mimesis          &   \textbf{78.82\%}    &   0.00\%   &   55.29\%        &    64.71\%   &  0.00\%\\
sanic            &   \textbf{38.95\%}    &  38.32\%   &   31.70\%        &   27.76\%   & 0.00\%\\
pytutils         &   34.87\%             &  25.89\%    &  \textbf{41.28\%}&   31.19\%   &  0.00\%\\ 
thonny           &   \textbf{42.62\%}      &  25.79\%    &  19.57\%         &   28.15\%  &  3.62\% \\ 
cookiecutter     &   \textbf{47.28\%}    &  25.00\%    &  25.25\%&   14.36\%   &  0.00\% \\ \hline
Avg.             &   \textbf{42.21\%}    &  24.98\%    &   32.72\%        &   31.67\%   & 5.76\% \\ \hline
\end{tabular}
\caption{Line Coverage Scores on Complex Methods.}
\label{tab:line}
\end{table}


\subsection{Execution Correctness Check}
We compare the execution correctness among the five unit test generation methods.
LLM-based unit test generation methods rely on the reasoning capabilities of LLMs to understand code and generate test cases. 
Compared with traditional search-based approaches, LLMs, with their probabilistic nature, are more likely to generate incorrect test cases.
Specifically, the number of passed test cases is calculated as the total number of test cases minus those that resulted in failures, errors, or were skipped.
The results of the pass rates are shown in Table~\ref{tab:pass}.

According to Table~\ref{tab:pass}, \textbf{the pass rate of our approach is not significantly higher than that of other methods.}
We also observe that a higher pass rate does not necessarily indicate a higher coverage rate.
This may be because, during the iterative generation process, the LLM tends to produce many similar unit test cases targeting the same code lines, resulting in limited additional coverage despite a high pass rate.
For example, the pass rate of our method on the project \textit{flutils} is 63.83\%, much higher than it is on project \textit{mimesis}, which is 24.45\%, while the line coverage rate is opposite.

We also observe that when Pynguin successfully generates executable test cases, its pass rate is generally higher than that of LLM-based test generation methods. 
This highlights a limitation of LLM-based approaches in this aspect and indicates an important direction for future improvement.


\renewcommand{\arraystretch}{1.2}
\begin{table}[t]
\centering
\small
\begin{tabular}{@{}lrrrrr@{}}
\hline
\textbf{Project} & \textbf{Our} & \textbf{CUT} & {\textbf{TELPA}} & \textbf{HITS} & \textbf{Pynguin}\\ \hline
dataclasses-json &   63.31\%    &     \textbf{64.28\%}     &   63.83\%     &    46.00\%  & 77.78\%  \\
docstring parser &   29.06\%    &      \textbf{40.00\%}   &   20.63\%        &    34.57\%  & 83.33\%  \\
flutes           &   \textbf{38.75\%}    &   21.15\%   &   27.91\%        &    24.44\% & N\//A  \\
flutils          &   \textbf{63.83\%}    &   38.72\%   &   63.32\%        &    46.00\% &  50.00\%  \\
mimesis          &   24.45\%    &   0.00\%   &   \textbf{42.85\%}       &    37.33\%   & N\//A \\
sanic            &   10.86\%    &  \textbf{35.58\%}   &   15.15\%        &   20.42\%  & N\//A  \\
pytutils         &   47.36\%    &   25.32\%        &    \textbf{51.85\%}      &    28.57\%   & N\//A  \\ 
thonny           &   \textbf{32.60\%}    &  19.27\%    &  11.94\%&   16.35\%   & 50.00\%  \\ 
cookiecutter     &   \textbf{29.00\%}    &  23.07\%    &  13.07\%&   17.12\%   & N\//A  \\  \hline
Avg.             &   \textbf{37.69\%}    &  29.71\%    &  34.50\%&   30.08\%   & 65.29\% \\ \hline
\end{tabular}
\caption{Pass Rate Comparison on Complex Methods.}
\label{tab:pass}
\vspace{-2mm}
\end{table}

\subsection{Ablation Study}
We analyze the components’ contribution to improving coverage scores through the ablation study.
Compared to directly using LLMs to generate unit test cases for the software under test, our approach incorporates several key techniques: (1) context information retrieval, which extracts project dependencies and summarizes them using the LLM; (2) code elimination, which removes code that has already been covered; and (3) iteration generation without elimination, in which the LLM is called multiple times to generate unit tests for a code slice, refining them based on coverage reports until an iteration limit is reached. 
We perform ablation studies to investigate the contribution of each component within our method.
The results are shown in Table.~\ref{tab:ablation}.

\renewcommand{\arraystretch}{1.2}
\begin{table}[!htbp]
\centering
\small
\begin{tabular}{lrrrrr}
\hline
\textbf{Project} & \textbf{Our} & \textbf{w/o E} & {\textbf{w/o I}} & \textbf{w/o D} \\ \hline
dataclasses-json &  \textbf{24.27\%}       &   14.56\%      &  16.99\%    &     20.54\%   \\
docstring parser &  \textbf{40.71\%}       &   25.33\%      &  16.54\%    &     36.54\%   \\
flutils          &  \textbf{41.04\%}       &   29.14\%      &  11.42\%    &     38.05\%   \\
mimesis          &  \textbf{78.82\%}      &    54.55\%      &   0.00\%    &     75.29\%  \\
cookiecutter     &  \textbf{42.62\%}      &    28.44\%      &  36.13\%    &     37.62\%   \\ \hline
Avg.             &  \textbf{45.49\%}      &    30.40\%      &  16.22\%    &     41.61\%   \\ \hline
\end{tabular}
\caption{Ablation Study. \textbf{w/o E} indicates without code elimination. \textbf{w/o I} means without iteration interact with LLMs. \textbf{w/o D} represents without dependency information.}
\label{tab:ablation}
\vspace{-2mm}
\end{table}


We compare three variants of our method: (1) ``w/o Elimination", which removes the elimination module from the system; (2) ``w/o Iteration", which sets the iteration limit in the LLM-based test generation without elimination module to 1; and (3) ``w/o Dependencies", which directly inputs the code slices without dependencies.

As shown in the table, the decrease in coverage score caused by removing the context information retrieval step (w/o Dependencies) is the smallest.
This may be because related code within the same file as the target unit is still included in the code slice, whereas the influence of external packages is relatively limited. 
Since these projects are high-quality open-source software, the LLM can often infer the usage of referenced functions and classes from detailed comments and naming conventions.

Both code elimination and iteration number have a significant impact on the system’s performance. 
However, removing iteration leads to an even larger decrease in coverage score. 
This is because the output of LLMs is highly unstable; without iteration, if the LLM fails to generate test cases that cover new lines, the method with code elimination will terminate immediately, rendering the elimination module ineffective.
\textbf{All steps and components in our method contribute to the performance improvements in the average coverage scores}. 


%% file: sections/7_conclusion.tex
\section{Conclusion}~\label{con}
We highlight the limitations of LLM-based unit test generation methods in achieving high coverage when applied to complex methods under test.
We follow the divide-and-conquer paradigm and propose a method to enhance the LLM-based unit test generation by iteratively eliminating covered code through static analysis.
We conducted a comprehensive evaluation of our method against state-of-the-art LLM-based and search-based methods, demonstrating the effectiveness of our approach.

%% file: sections/8_appendix.tex
\appendix
\section{An Illustrative Example}
Here we provide an illustrative example of code elimination process, as shown in Fig~\ref{fig:example}.
\begin{figure*}[htbp]
    \centering
    \includegraphics[width=\textwidth]{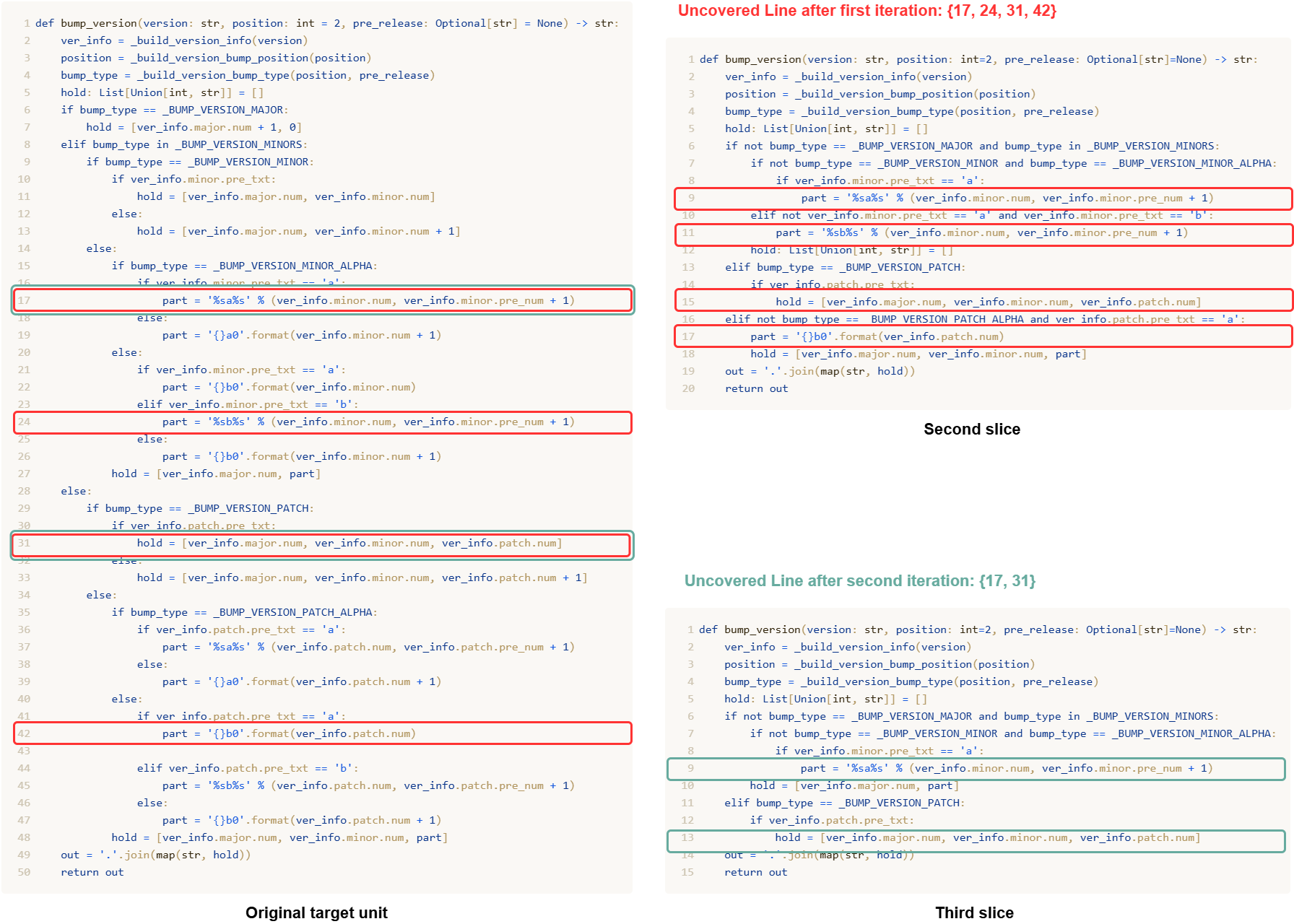}
    \caption{An Example of Eliminating Covered Code.}
    \label{fig:example}
\end{figure*}

On the left is the original target unit, which is a complex function \textit{bump\_version($\cdot$)} from an open-source Python project \textit{flutils}.
As can be seen, this target unit contains 50 lines of code and includes multiple layers of nested conditional statements, making it challenging for the LLM to achieve full coverage in a single attempt.
The code lines highlighted with red boxes, which are Lines 17, 24, 31, and 42, represent the uncovered lines identified after the unit test generation without elimination.
Based on these uncovered lines, we apply Algorithm~\ref{alg:alg1} to the original target code to perform code elimination.
This process yields the refined second code slice shown in the upper-right corner.
As shown, we preserve all predecessor code lines of the uncovered lines.
For example, Lines 2, 3, 4, and 5.
This ensures the execution logic for the uncovered line remains unchanged.
Additionally, we address structural issues that may arise when branches under an \textit{if} statement are eliminated, which could otherwise invalidate the corresponding \textit{else} or \textit{elif} blocks.
For instance, in the second slice, Lines 6 and 7 are changed to maintain syntactic and semantic correctness.
At this point, the code has been reduced to 20 lines, with all unnecessary statements removed.

After unit test generation without elimination for the second slice, a new set of uncovered lines is identified.
Specifically, they are Lines 17 and 31 in the original target unit.
Our method then applies the code elimination process to the target unit again based on these newly uncovered lines.
This results in a more compact third slice, consisting of 15 lines, as shown in the lower-right corner.

This example demonstrates that our approach effectively decomposes the problem of generating unit test cases for a complex target function into generating tests for a series of smaller code slices.
By doing so, it addresses the scalability issue faced by LLM-based unit test generation methods in a principled and efficient manner.

\section{Experiment Details}
\subsection{Baseline Implementation}
For TELPA~\cite{yang2024enhancing}, we re-implement the approach based on the description provided, as the original paper does not release the source code.
This baseline adopts an architecture similar to the LLM-based test generation without elimination in our method. The key difference lies in its iterative process: rather than terminating upon detecting newly covered lines after test validation, it continues the generation process until either full coverage is achieved or a predefined iteration limit is reached.
The iteration limit for refinement is set to 5.
For prompts used in refinement, we construct them using the backwards analysis approach proposed in the original work.
For ChatUniTest~\cite{chen2024chatunitest}, the original code is designed for Java, and we implement our own version.
We first prompt the LLM to generate unit tests using the same initial prompt as in our method, and then iteratively invoke the LLM to perform bug fixes on the generated tests. 
The maximum number of bug-fixing iterations is set to five.
For HITS~\cite{wang2024hits}, the original code is also designed for Java, and we implement our own version.
We first leverage the LLM to segment the target unit into several segments and summarize the functionality of each segment.
Then, the LLM is prompted to generate unit tests for each segment individually.
Finally, the LLm is used to perform bug fixing similar to ChatUniTest, with the number of iterations limited to five.

\subsection{Experiment Setup}
We use OpenAI's GPT-4o~\cite{openai2024gpt4o} as the large language model in this work since this model strikes a good balance between the price and the performance in code-related works.
For the GPT-4o API, we use the default values for the temperature and token length limit parameters. 
Specifically, the temperature is set to 1.0, and the token length limit is 8096.
For each LLM-based method, the number of interaction rounds with LLM in a dialogue is limited to 5.
It is worth noting that if the accumulated multi-turn dialogue context causes the prompt to exceed the token limit during an LLM invocation, we treat this invocation as a failure and terminate the loop.
This applies to both our method and TELPA's multi-turn refinement of generated unit tests, as well as ChatUniTest and HITS's multi-turn bug-fixing processes.
During the experiments, we select target methods or functions from nine open-source projects based on the criterion of having more than 50 lines of code.
We also collect dependency information for each target unit. 
All relevant information for each target unit is stored in a corresponding JSON file.
The datasets will be released upon publication.
During evaluation, each method is tested using its JSON file directory as input to generate and assess the corresponding unit test. 
It is important to note that, since different baselines originally adopt their own strategies for incorporating dependency information into the initial prompt, we standardize this step by using the dependency collection method proposed in this paper.

\section{Discussion}
\subsection{Future Improvement}
The most straightforward enhancement is to use a more powerful LLM. 
In our experiments, we selected the widely used GPT-4o for its balance between performance and cost-effectiveness. 
However, it is easy to substitute this model with more advanced models that are specifically optimized for code-related tasks, such as DeepSeek-Coder~\cite{deepseek-coder}, GPT-o3~\cite{openai2024gpto3}.
These models can be seamlessly integrated into our framework to potentially further improve unit test generation quality.

In our implementation, we deliberately avoided extensive prompt engineering. 
The initial prompt consists of the code and its dependency file, accompanied by a few simple examples as guidance. 
During the iterative refinement process, we simply use uncovered lines or run-time errors as prompts to instruct the LLM to further improve the generated test cases.
In fact, users can easily incorporate their own prompt engineering strategies into our method.
For example, prompts constructed through static analysis as used in TELPA can be integrated to further guide the LLM in addressing uncovered paths.
Meanwhile, other techniques such as Retrieval-Augmented Generation (RAG)~\cite{gao2023retrieval} can also be integrated into the system's iterative test generation module—specifically the version without code elimination—to further enhance its effectiveness.

\subsection{Threats to Validity}
The main threats to the validity of our method lie in the limited dataset and its generalizability to other LLMs.
Due to budget constraints, our experiments were conducted on complex methods from only nine projects. 
Moreover, the inherent variability in LLM performance may affect the magnitude of the differences in coverage rates among the unit tests generated by different methods in our experiments.
Nevertheless, the results sufficiently demonstrate the superiority of our approach in addressing scalability challenges. 
However, for real-world deployment, larger-scale testing would be necessary to uncover potential issues within our codebase. 
In our experiments, we only select the widely used GPT-4o model, which offers a favorable trade-off between generation capability and API cost. 
Expanding the dataset and evaluating additional LLMs remain important directions for future work.

Meanwhile, there are still areas for improvement in our implementation. 
For instance, although we explicitly instruct the LLM via prompts to place executable unit test code within the \verb|<answer>...</answer>| structure, the inherent instability of LLM outputs often leads to unexpected formatting issues—such as the inclusion of illegal string patterns like triple backticks (```), which can disrupt code execution. 
Additionally, various path dependency issues may cause the generated tests to fail at run-time. 
These factors can affect the stability of the resulting coverage rates. 
In future work, we plan to further refine the process of extracting test cases from LLM outputs to improve robustness and reliability.